\documentclass[runningheads,a4paper]{llncs}

\usepackage{cite}
\usepackage{amsmath,amssymb,amsfonts}
\usepackage{algorithm}
\usepackage{algpseudocode}
\usepackage{graphicx}
\usepackage{textcomp}
\usepackage{xcolor}
\usepackage{multicol}
\usepackage{verbatim}

\begin{document}

\title{The Bouquet Algorithm for Model Checking Unbounded Until}

\author{Shiraj Arora \and M. V. Panduranga Rao }

\institute{Indian Institute of Technology Hyderabad\\
India \\
\{cs14resch11010, mvp\}@ iith.ac.in}

\maketitle

\begin{abstract}
The problem of verifying the ``Unbounded Until" fragment in temporal logic formulas has been studied extensively in the past, especially in the context
of statistical model checking.
Statistical model checking, a computationally inexpensive sampling based alternative to the more expensive numerical model checking technique,
presents the following decision dilemma--what length of the sample is enough in general?

In this paper, we discuss an algorithm for this problem that combines ideas from graph theory, statistical model checking and numerical model checking.
 We analyze the algorithm,  
and show through experiments that this approach outperforms the standard statistical model checking 
algorithm for verifying unbounded until for low density Discrete Time Markov Chains.
\end{abstract}


\section{Introduction}
Probabilistic model checking deals with algorithmic verification of properties desired of stochastic systems.
One useful formalism for modeling such systems, with which we will be concerned in the present work is the discrete time Markov chain (DTMC).
Properties to be verified are formally specified as formulas in temporal 
logics such as PCTL~\cite{HanssonJ94}. There are primarily two techniques to perform probabilistic model checking. 
Numerical model checking computes the exact solution, albeit at a prohibitive cost due to the state space explosion in the underlying model. 
On the other hand, sampling based statistical techniques works by executing finite-length 
runs of the DTMC and evaluating the temporal logic formula on each run. These techniques offer a trade-off between the desired accuracy and
time, in terms of the number of samples generated for analysis.

Most temporal logics contain path operators called \emph{bounded} and/or \emph{unbounded until} operators.
When the formula contains only bounded until operators, the length of the path to be sampled
can be made to depend on the time bounds present in these operators. However, for the unbounded until operator, we face the
dilemma of choosing an appropriate length for the path to be sampled. The problem of verifying properties with unbounded 
until operators has been explored using several approaches~\cite{roohi2015,agha2018}. 

In this paper, we utilize the graph structure underlying the DTMC to address this problem. Our target applications are those where the system has a stable
description--the DTMC is fixed--and an extended analysis needs to be carried out through several queries over a course of time. Crucial to our approach
is the question of whether or not there are at least $k$ number of nodes reachable from any given node of the DTMC. This information can either be pre-annotated
or can be progressively annotated as formula evaluations proceed. Indeed, this annotated DTMC has to be stored across formula evaluations (as long as the system needs to be analyzed) for fully benefiting 
from the annotations done while evaluating previous formulas. The benefits in terms of speed is maximized when the graph is completely annotated.

The salient contributions of this work are twofold: (a) we propose a new algorithm, which we call the bouquet algorithm, for dealing with the 
unbounded until dilemma in statistical model 
checking using the structure of the graph underlying the DTMC, sampling and numerical model checking.  To the best of our knowledge, this approach is 
new, and (b) we show improved performance for low density DTMCs. Indeed, we give empirical evidence that in the case of completely annotated DTMCs,
the bouquet algorithm outperforms the standard statistical model checking algorithm.

Clearly, this approach suffers from the disadvantage of having to generate the entire DTMC explicitly, as do some other approaches reported in literature.
In such cases, extremely large DTMCs may need to be stored on a slower memory drive and consequently, there is a significant overhead of I/O operations 
while performing statistical model checking. We argue that the bouquet algorithm cuts this expense down.

The paper is arranged as follows: Section \ref{preliminaries} introduces some preliminaries and discusses previous work, 
section \ref{bouquet} describes the bouquet algorithm and analyzes the performance. 
Section \ref{experiments} discusses experimental results for low density DTMCs. Section \ref{conclusion} concludes the paper.

\section{Preliminaries and Previous Work}\label{preliminaries}

A Discrete Time Markov Chain (DTMC) is a Markov process defined in discrete time and 
described as a  tuple $M : (S, s_{init},\mathbb{P},AP,L) $ where \begin{itemize}
\item $S$ is a finite non-empty set 
of states
\item $\mathbb{P}:S\times S \rightarrow [0,1]$ gives the transition probability 
between two states in $S$ such that $\forall s\in S : \sum_{s'\in S} \mathbb{P}(s,s') = 1$,
\item $s_{init} \in S$ is the initial state, (in general it is a probability distribution over a subset $S_{init}$ of initial states.
We restrict our discussion to the case of a single initial state, for simplicity.)
\item $AP$ is a set of atomic propositions and 
\item $L:S \rightarrow2^{AP}$ is a labeling function.
\end{itemize}
Fig.~\ref{dtmc} illustrates an example DTMC with state space $S=\{0,1,2,3,4,5\}$ and atomic propositions 
$AP=\{p,q,r\}$. 

\begin{figure}
    \centering
    \includegraphics[width=0.6\linewidth]{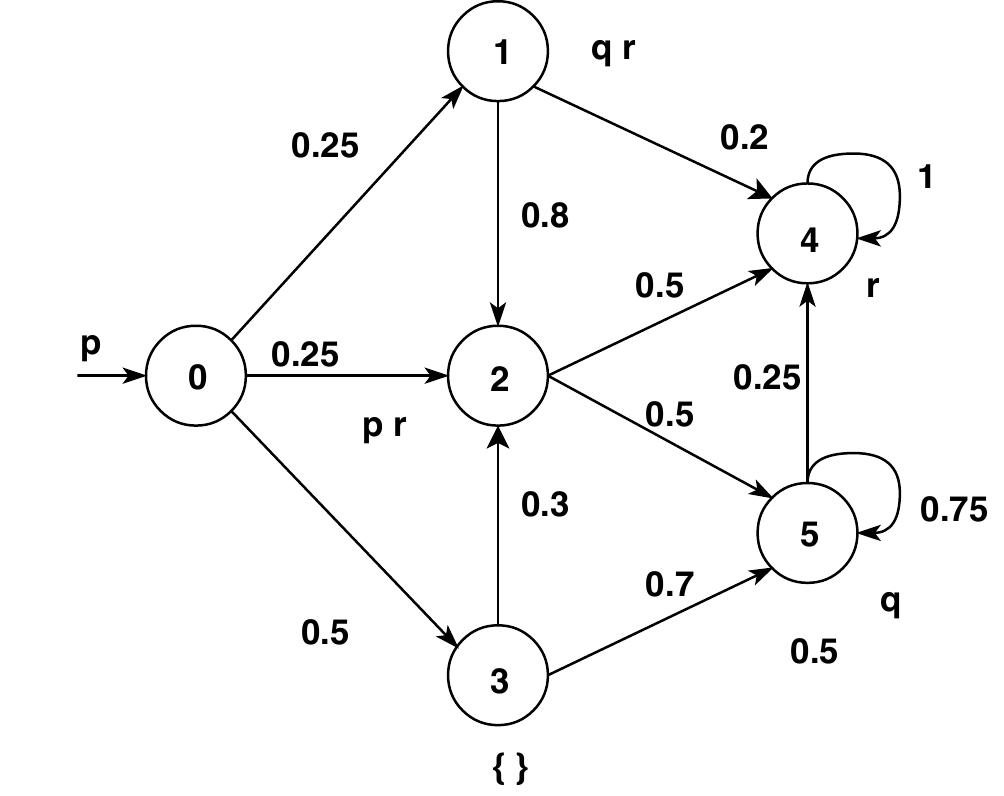}
    \caption{An example DTMC}
    \label{dtmc}
\end{figure}

A path $\pi$ in a DTMC M is a sequence of states $s_0, s_1, s_2 \ldots$
such that for all $i = 0, 1, 2, \ldots s_i \in S$ and $\mathbb{P}(s_i,s_{i+1}) > 0$. The $i + 1$th state in a
path $\pi$ is denoted by $\pi[i]$. We use the terms paths, \emph{samples} and \emph{traces} interchangeably in this paper.
%The set $Path(s)$ is the set of all infinite paths starting
%from state $s$ in $M$.

Given a labeling of atomic propositions to a state $s$, we can talk of boolean formulas constructed from atomic propositions and boolean connectives 
(OR and not). These formulas are evaluated per state. If the atomic proposition assignment at a state $s$ results in such a boolean \emph{state} formula $\Phi$ 
being evaluated to true, we write $s\vDash \Phi$.

The \emph{path} formula \emph{unbounded until}, an important fragment of temporal logics like PCTL and LTL, is written as: 
$\psi\ ::= \ \Phi_{1} U\ \Phi_{2}$, where $\Phi_1$ and $\Phi_2$ are state formulas as defined above.
The semantics of the unbounded until is straightforward:
a path $\pi$ satisfies $\Phi_1 U \Phi_2$, written as $\pi\vDash \Phi_1 U \Phi_2$ iff 
	\[ \exists i \geq 0\ |\ \pi[i] \vDash \Phi_{2} \land 
	\ \forall j < i, \pi[j]\vDash \Phi_{1}\] 

The bounded version of this operator, denoted  $\psi\ ::= \ \Phi_{1} U^{\leq t}\ \Phi_{2}$ has the following semantics:
	\[ \pi \vDash \Phi_{1} U^{\leq t} \Phi_{2} \textrm{ iff } \exists i \geq 0 \textrm{ and } \leq t \mid \ \pi[i] \vDash \Phi_{2} \land 
	\ \forall j < i, \pi[j]\vDash \Phi_{1}\] 

\subsection{Probabilistic Model Checking}
Two major techniques for (probabilistic) model checking of stochastic systems are numerical model checking and statistical model checking. We now briefly
discuss salient aspects of these two techniques. 
	Numerical Model checking~\cite{courcoubetis88,HanssonJ94,ModelCheck} is a graph reachability based verification technique
	that computes the exact probability of a (e.g PCTL) formula being true. The crux of the technique is to perform a bottom-up traversal of the syntax tree of the PCTL state formula and 
	identify the states that satisfy the formula. Numerical model checking is useful in verifying a variety of 
	stochastic systems that can be modeled as discrete and continuous time Markov chains, Markov Decision Processes, etc. 
The computation cost of model checking discrete-time 
	Markov chains using these techniques is polynomial in terms of the
	size of the input model as well as input property. However, this approach
	becomes expensive for large systems that suffer from state space explosion.
		
Statistical 
	model checking is a sampling based technique which executes multiple runs of the input system. Statistical model
	checking algorithms are mainly of two categories: Hypothesis testing based~\cite{sen2004,Younes,sen05} and Estimation based~\cite{hoeffding1994,herault2004,nimal2010}. While estimation 
	based algorithms seek to calculate the probability of satisfying a given property with some loss in accuracy, hypothesis testing 
	based algorithms check if the probability meets the required threshold or not. 
	Statistical Model Checking~\cite{sen2004,Younes,sen05} is a faster alternative to verify the property of the input system at the cost of 
	accuracy. 
	Younes et al.~\cite{NumVsStat} provide a detailed comparison between numerical and statistical algorithms to verify the temporal properties.  

In this paper, we will focus on estimation based statistical model checking. The Statistical Model Checking (SMC) algorithm proceeds as 
follows: A sample trace $\pi_i$ of a maximum length $maxPathLength$ is generated and assigned a value $b_i=1$ if the unbounded until
formula is satisfied in $\pi_i$ and $b_i=0$ otherwise. The probability estimate $p'$ of satisfying the unbounded until formula after generating $N$
sample traces of the DTMC is: 
\[p'=\frac{1}{N} \sum_{i=1}^n b_i\]
The Chernoff-Hoeffding bound~\cite{hoeffding1994,herault2004}
is then used to compute the number of samples needed to estimate the resulting probability with a desired accuracy. If $p$ is the actual
probability of the formula being satisfied. Then, to achieve $Prob[|p-p'|\leq \epsilon]\geq 1-\delta$ for $\epsilon,\delta>0$, then the
Chernoff-Hoeffding bound requires that the number of  
samples $N$ needed is given by  
\begin{equation}
        N \geq \frac{ln(\frac{2}{\delta})}{2\epsilon^2}
\end{equation}

	One problem that is encountered in statistically verifying a logic that contains the unbounded until fragment
is the dilemma of when to stop a run. A simple option is to set a limit on the length
of the run. If the formula is not evaluated conclusively (either true or false) before this length, it is classified as a false. This can  
potentially result in a loss in accuracy.
On the other hand, the bounded until offers a natural bound on the length of a sample run for the formula to be evaluated conclusively.
	%In this paper, we will focus on estimation based statistical model checking algorithms and particularly with be working with Chernoff-Hoeffding bounds~\cite{hoeffding1994,herault2004}
%	to compute the number of samples needed to estimate the resulting probability with given approximation parameter $\epsilon$ and confidence interval $\delta$. 
	
One of the first attempts to statistically verify unbounded until properties was by 
Sen et al.~\cite{sen05} where they introduced the notion of stopping probability $p_s$--at every state $s$ in a path there exists probability 
$p_s$ with which the 
generation of the trace terminates at the current state $s$. They 
also identify the set of states for which the probability of satisfying the 
unbounded until formula is zero, through sampling. They estimate the probability of the formula being true using Bernoulli trials.
Younes et al~\cite{younes2010} extend the 
concept of stopping probability by using non-Bernoulli trials 
to estimate the probability. While the stopping probability depends on the
size of the model in~\cite{sen05}, it depends on the subdominant eigenvalue of the 
transition matrix of the model in~\cite{younes2010}. 
Younes et. al also propose an algorithm in~\cite{younes2010} to identify the
set of states with zero probability of satisfying the unbounded until formula using 
reachability analysis. 

Rabih and Pekergin~\cite{rabih2009} and Lassaigne and Peyronnet~\cite{lassaigne2015} use the subdominant eigenvalue 
of the transition matrix for an ergodic Markov chain to estimate the upper 
bound on length of path for unbounded until. They then solve the unbounded 
until formula as a bounded until formula, with this estimate as the bound.
Basu et al.~\cite{basu2009} and He et al.~\cite{he2010} also convert the unbounded until to a bounded until 
formula by selecting an arbitrarily large time bound for the bounded operator
such that the resulting probability is same for both the formulas.

None of the above approaches take into account the structure of the Markov Chain. Daca et al. \cite{daca2017} 
proposed an algorithm that utilizes the minimum transition probability of the Markov chain
to identify the probable bottom  strongly connected components (BSCC) in a Markov chain. 
They execute sample runs only till they reach one of the states in a BSCC. 
This algorithm is the closest to the bouquet algorithm reported in this work.

\section {The Bouquet Algorithm}\label{bouquet}

The bouquet algorithm is essentially a hybrid algorithm that combines statistical and numerical model checking.
%In the underlying directed graph of the Markov chain, we identify the states that 
%have atmost $k$ reachable states. 
%We refer to these states as \textbf{\textit{flower}} states.  
The algorithm begins by sampling a trace in the DTMC as in the case of statistical model checking.
A trace of the DTMC $M$ is generated, starting from the input state, until it either satisfies (or rejects) 
the unbounded until formula, or a state $s_F$ is reached from which there are at most $k-1$ reachable states. 
In case of former, the result of the corresponding trace is evaluated to True (or False). In the latter case, 
we isolate the state $s_F$ and states reachable from $s$ to form an induced DTMC $M'$ with
at most $k$ states and $s_F$ as the initial state. The algorithm performs a numerical model checking procedure on 
$M'$.

%We then verify the same unbounded until formula
%for $M'$ using numerical model checking algorithm. The calculated probability is stored as 
%the result for the corresponding trace as well as the state $s$. This resulting probability is reused 
%for any subsequent visits to the state $s$. 

%We know that NMC outperforms SMC for small sized Markov Chains. Thus, in this paper, we utilise 
%the benefits of both NMC algorithms for Markov chains with smaller sized state space,
%and SMC algorithms for Markov chains with large state space.

We refer to the induced Markov chain $M'$ as the \textit{flower} $F$ rooted at $s_F$. 
Note that rooted at any state $s_F$ in a trace, there can be at most one flower.
A stalk for a flower $M'$ in $M$ is the finite-length path from the initial state in $M$
to $s_F$. 
A \emph{bouquet} $B$ is a set of $(flower,stalk)$ tuples. In a sense, generating the bouquet constitutes the 
bouquet algorithm, and hence the choice of the name.

Figure \ref{flower} shows part of the underlying directed graph of 
an input Markov chain $M$. If the value of $k$ is 4 for $M$ then $M_1$ and   
$M_2$ are flowers in $M$ with $s_{11}$ and $s_{15}$ as initial states. 
Thus, $s_0,s_2,s_6,s_1,s_6,s_{11}$ is a stalk for flower $M_1$.

\begin{figure}
    \centering
    \includegraphics[width=0.6\linewidth]{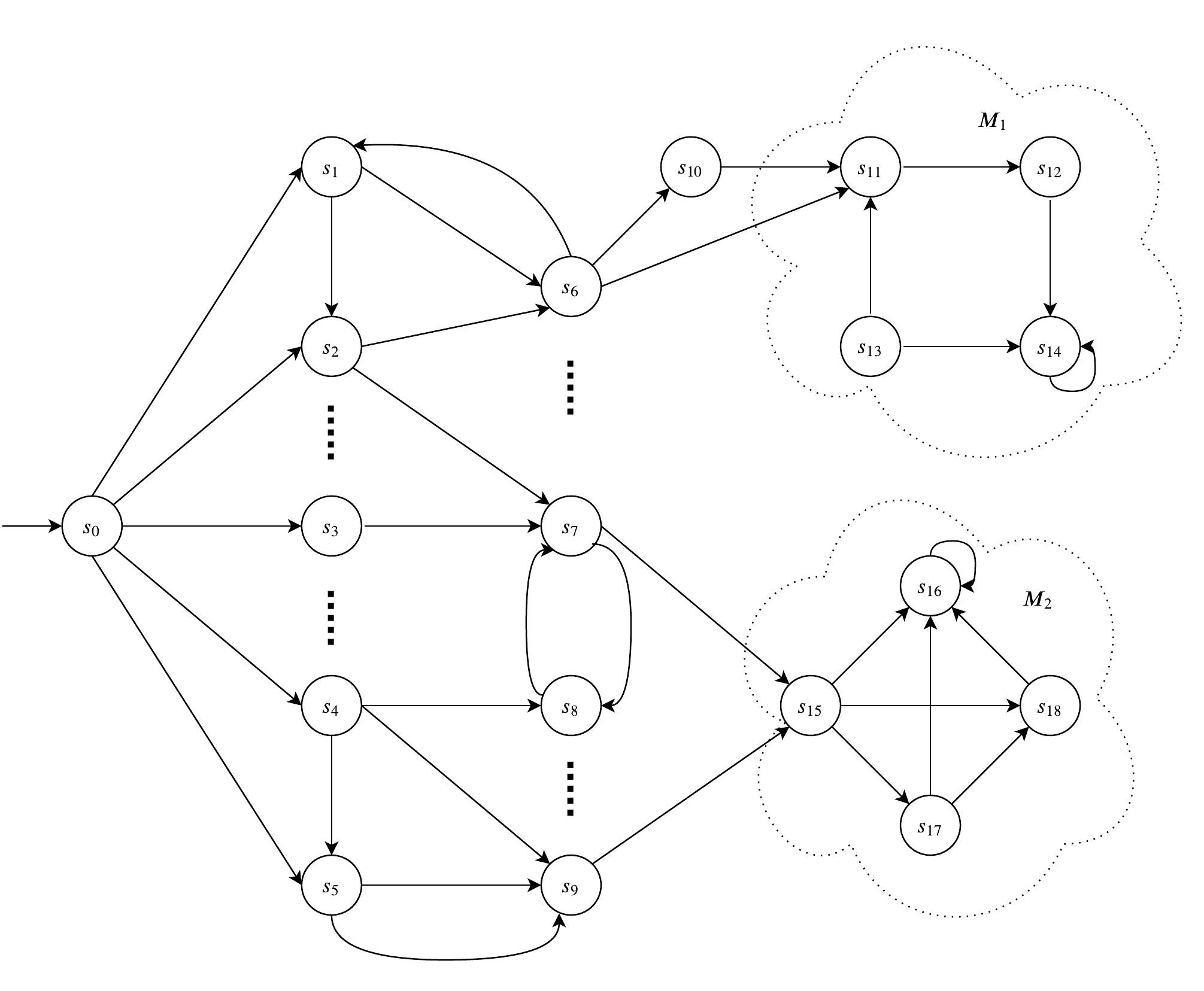}
    \caption{Example of directed graph $G$ for the input Markov chain $M$}
    \label{flower}
\end{figure}

%%%%%%%%%%%%%%%%%%%%% Subsection Details%%%%%%%%%%%%%%%%%%%%%%%%%%%%%%%%%%%%%
\subsection{Details}
\algrenewcommand\algorithmicindent{0.75em}%
\begin{algorithm}
	\caption{Bouquet Algorithm }
	\label{algorithm}
	\begin{multicols}{2}
	\begin{algorithmic}{\small}
		\State \textbf{Function:} Bouquet($M, \Phi,k, N_B,rProb$)
		\State $res \gets 0$
		\For{$i=0$ to $N_B$}
			\State $c\_s \gets s_{init}$; $step \gets 0$
			\State array $A \gets \emptyset$
			\While{$step\leq maxPathLength$}
				\If{$c\_s \vDash b$}
					\State $res \gets res +1$
				\ElsIf{$c\_s \vDash \lnot a \land \lnot b$}
					\State break
				\ElsIf{$c\_s.isAnnotated()$}
					\If{$c\_s.annotationValue$}
						\State $res \gets res+ getNMCresult(c\_s)$
					\EndIf
				\ElsIf{$random(0,1) \leq rProb$}
					\If{$isFlower(M,c\_s,k)$}
						\State $c\_s \gets findFlowerhead(M,c\_s,k,A)$
						\State $A \gets \emptyset$
						\State $M' \gets getFlower(M,c\_s,k)$
						\State $res \gets res+ doNMC(M',c\_s,\Phi)$
					\Else
						\ForAll{$s \in A$}
							\State $s.annotationValue=False$
						\EndFor
						\State $A \gets \emptyset$
					\EndIf 
				\Else
					\State $A.add(c\_s)$
					\State $step\gets step +1$
					\State $c\_s \gets getNextState(c\_s)$
				\EndIf
			\EndWhile
		\EndFor
		\State $res \gets res/N_B$
		\State \Return $res$
		\\\\ \\\\ \\\\
		
		\State \textbf{Function:} isFlower($M, s, k$)
		\If{$s.isAnnotated()$}
			\State \Return $s.annotationValue$
		\EndIf	
		\State $list , stack \gets neighbours(s)$
		\State $c \gets size(list)$
		\While{$c <k \land \lnot stack.isEmpty() $}
			\State $t \gets stack.pop()$
			\ForAll {$i \in neighbours(t)$}
				\If{$\lnot list.contains(i)$}
					\State $c \gets c +1$
					\State $list.add(i)$ ; $stack.push(i)$
				\EndIf
			\EndFor
		\EndWhile
		
		\If{$c<k$}
			\State \Return $True$
		\Else
			\State \Return $False$
		\EndIf
		\\
		\State \textbf{Function:} findFlowerhead($M,c\_s,k, A$)
		\State $l \gets 0$;$h \gets A.size()$
		\State $flowerhead \gets c\_s$
		\While{$l \leq h$}
			\State $s \gets A[\lfloor\frac{l+h}{2}\rfloor]$
			\If{isFlower(M,s,k)}
				\For{$i=\lfloor\frac{l+h}{2}\rfloor$ to $h-1$}
					\State $A[i].annotationValue=True$
				\EndFor
				\State $h \gets \lfloor\frac{l+h}{2}\rfloor-1$
				\State $flowerhead \gets s$
			\Else
				\For{$i=l$ to $\lfloor\frac{l+h}{2}\rfloor$}
					\State $A[i].annotationValue=False$
				\EndFor
				\State $l \gets \lfloor\frac{l+h}{2}\rfloor+1$
			\EndIf
		\EndWhile
		\State \Return $flowerhead$
	\end{algorithmic}
	\end{multicols}
\end{algorithm}

We now describe the bouquet algorithm in detail.
Given an input DTMC $M=(S,s_{init},\mathbb{P},AP,L)$ with $n=|S|$ as the 
number of states in $M$  and an input (say, PCTL) formula $\Phi= Pr_{=?}[a\ U \ b]$ with $a,b \in AP$,
the Bouquet algorithm described in Algorithm \ref{algorithm} 
estimates the probability of state $s$ satisfying the formula $\Phi$.
The algorithm also takes as input the total number of sample runs it needs to 
execute, $N_B$, the size of the flower $k$ for the model and a probability $rProb$ of searching for a flower in $M$.
We use Chernoff-Hoeffding bound~\cite{hoeffding1994,herault2004} to calculate 
the number of required samples. The number of samples needed for SMC for given 
approximation parameter $\epsilon$ and confidence $\delta$ is 
$N_s \geq \frac{ln(\frac{2}{\delta})}{2\epsilon^2}$. We empirically decide the value of $N_B$ 
as a fraction of $N_s$ such that same approximation and confidence values is achieved.
We will discuss how to fix $k$ in the next subsection.
 
The bouquet algorithm begins by sampling traces from the input Markov chain $M$, 
as in the case of SMC.
At every execution step of a trace, we first check if the unbounded until formula 
evaluates to either true or false at the current state $s$. If so, this trace is deemed successful, and not continued further.
Otherwise, before visiting the
next state in the trace, we check with a probability $rProb$ if a flower is present at $s$. 
If no flower exists at $s$ or if a flower is not searched for at $s$, then the next state in the trace is traversed.
However, if a flower is encountered at $s=s_F$, then the bouquet algorithm computes 
the exact probability of $s_F$ satisfying the unbounded until formula, through numerical model checking and 
annotates this probability in the  DTMC $M$. This allows 
re-usability of previously computed results--if another trace visits the state $s$ later, then the computed probability can be directly used, 
instead of generating additional traces as in the case of SMC.  

The algorithm for identifying a flower $M'$ in $M$ is described in the $isFlower$ 
function of Algorithm \ref{algorithm}. The function first checks if a 
reachability computation has already been done on the current state $s$. If so, 
it directly returns the result. Otherwise, it identifies the number of reachable states 
from $s$. If this number is small enough, that is, less than $k$, then we extract 
these states from $M$ to create another DTMC $M'$ with $s_F$ as the 
initial state.  $M'$ is simply the DTMC corresponding to the subgraph induced by the vertices reachable from $s_F$.
If the number of reachable states from $s$ is not less than $k$, 
then it returns False. 
Then, we continue with the execution of the trace until either a conclusive result or a flower is found.   

The function $isFlower$ is called with a probability of $rProb$ at every state in trace generated by the bouquet algorithm.
Thus, with probability $1-rProb$, we skip the step of searching for a flower at the current state. We add into an array $A$
such states which are traversed in a trace but exempted from flower search. Whenever a flower is found
in the future, the bouquet algorithm calls the function $findFlowerhead$ to annotate the states in the array $A$
and find a possibly larger flower of size less than $k$. The function $findFlowerhead$ uses binary search 
to identify this possibly larger flower.

\subsection{Fixing $k$}

We desire to fix the size $k$ of the flower through the following analysis.
Note that this only provides a heuristic. By relaxing or tightening some of the 
assumptions depending on the underlying DTMC, one can arrive at a different $k$.

Let $n$ be the number of vertices in the underlying directed graph $G$ 
of Markov Chain $M$ and let the density of the graph $G$ be $\rho$ such that 
$G$ is a sparse graph:
\[\rho = \frac{\textrm{No. of edges in G}}{n(n-1)}<<1.\] 

In what follows, we assume that the $\rho$ is uniform across the graph--for any induced subgraph with at least $k$ nodes, the density remains
unchanged.
Let $G'$ be the underlying directed graph of the flower Markov chain $M'$.
Then by the definition of the flower $M'$, number of vertices in $G'$ is at most $k$. Then,
\begin{equation}
	\textrm{No. of edges in G'} = \rho k(k-1)
\end{equation}
If size of graph $G'$ is $k$, then $k-1$ vertices are reachable from the initial vertex in $G'$.
Thus, there exists at least $k$ edges in $G'$ ($k-1$ to ensure reachability and $1$ to ensure stochastic
property of the Markov Chain $M'$).  Thus, the graph $G'$ consists of at most $k$ vertices and at least $k$ edges. 
Then, 
\begin{equation}
	\rho k > 1
\end{equation}
Thus $\rho k-1>0$.
While the Bouquet algorithm calculates the probability of $M'$ satisfying the unbounded until 
formula using numerical model checking algorithm, statistical model checking algorithm would have to generate
multiple traces in $M'$, for comparable accuracy. 
We now estimate a lower bound on the number $r$ of unique traces of length $L$ that the SMC algorithm would have to sample.
To begin with, note that there exists at least one path connecting all the $k$ states in $M'$. 
Further, addition of one edge in $G'$ leads to addition of at least one unique path of length $L$ 
in $M'$. Since $G'$ contains $\rho k(k-1)$ edges, then there exists at least $\rho k(k-1) - k + 1 $ paths in the $M'$.
Thus,
\begin{equation}
	r \geq (\rho k -1)(k-1)
\end{equation}  
%Now, for numerical model checking of the unbounded until formula for the flower $M'$ in $M$, there exists at least $r$ 
%unique paths in $M'$ to verify the same formula statistically. 

%%%%%%%%%%%%%%%%%%%%%%%%%%%%%%%%%%%%%%%%%%%%%%%%%%%%%%%%%%%%%%%%%%%%%%%%
The cost of identifying and verifying a flower $M'$ of size $k$ using bouquet algorithm is:
\begin{equation}
\textrm{Cost}_B= O{(k^2)} + O(k^3)  +O(kn) = c_1 k^3 +c_2 k^2 + c_3kn +c_4
\end{equation}
The terms $O(k^2)$ and $O(k^3)$ correspond to the cost for the  precomputation steps and matrix multiplication
in NMC for a model of size $k$ respectively and $O(kn)$ for reachability search in a sparse graph.

Similarly, if $c_s$ is the cost any statistical model checking algorithm spends on a trace with maximum allowed length,
 then the computation cost of verifying a flower $M'$ using the  SMC algorithm is at most $c_sr$ where
 $r$ is the number of unique paths.
\begin{equation*}
        \begin{array}{ccc}
                \textrm{Cost}_S & \geq & c_sr = c_s(\rho'k -1)(k-1).
        \end{array}
\end{equation*}
Thus, $\textrm{Cost}_S = c_s(\rho k-1)(k-1)+c_5$.

We choose the size of the flower $k$ such that
\begin{equation*}
        \begin{array}{ccc}
                \textrm{Cost}_S & \geq &\textrm{Cost}_B \\
                c_s (\rho k -1)(k-1) +c_5 & \geq &c_1 k^3 + c_2k^2 +c_3 kn +c_4
                \end{array}
\end{equation*}
           \[     c_1k^3 + (c_2-c_s\rho) k^2 + (c_3 n +c_s(1+\rho)) k +c_4-c_s-c_5  \leq 0\]
Substituting with $C_1$ for $(c_2-c_s\rho)/c_1$, $C_2$ for $c_3/c_1$, $C_3$ for $c_s(1+\rho)/c_1$ and $C_4$ for $(c_4-c_s-c_5)/c_1$, we get
\begin{equation}
                k^3 + C_1 k^2 + (C_2 n +C_3) k +C_4 \leq 0  
        \end{equation}

Solving, we get $k\approx \sqrt{n}$. Indeed, we use $k=\sqrt{n}$ for the experiments reported in section 4.

%%%%%%%%%%%%%%%%%%%%%%%%%%%%%%%%%%%%%%%%%%%%%%%%%%%%%%%%%%%%%%%%%%%%%%%%
It is important to note that if a flower is not encountered, then both statistical algorithms and the bouquet algorithm will end up traversing 
upto $maxPathLength$. 

\subsection{Correctness}
We now show that for the same number of samples generated at the starting state of a DTMC $M$, the Bouquet algorithm is at least
as accurate as the $SMC$ algorithm.

\begin{theorem}
Let $p=Prob[aUb]$. 
For $N_B\geq \frac{\ln\frac{2}{\delta}}{2\epsilon^2}$ samples and $\epsilon, \delta>0$ chosen as in the SMC algorithm, 
let $p'$ be the approximation computed by the Bouquet algorithm of $p$. Then, $Prob[|p-p'|\leq\epsilon]\geq 1-\delta$.
\end{theorem}

\begin{proof}
Consider the standard SMC algorithm that generates $N$ samples. We argue that for these many samples, the Bouquet algorithm is at least
as accurate as the SMC algorithm.

In the context of the Bouquet algorithm, the traces generated by the SMC algorithm can be partitioned into the following classes:
\begin{enumerate}
\item Traces that are longer than $maxPathLength$: both algorithms behave identically.
\item Traces where $aUb$ is evaluated to True or False before the Bouquet algorithm encounters a flower: again, both algorithms behave identically.
\item Traces where the Bouquet algorithm hits a flower that has already been evaluated and annotated: the bouquet algorithm stops and reports the
probability. The SMC algorithm continues generation of the sample trace.
\item Traces where the Bouquet algorithm encounters an un-annotated flower $F$ before $aUb$ is evaluated to True or False:

Let the state where the trace encounters the flower be $s_F$. Then, the Bouquet algorithm calculates $Prob[aUb]$ exactly for the flower with initial
state $s_F$ using NMC. On the other hand,  
for this fragment of the DTMC, the SMC algorithm approximates $Prob[aUb]$ as
$p_{SMC}= \frac{\sum_r b_{F,r}}{r}$, where $r$ indexes the traces starting at $s_F$ into the flower $F$, explored by the SMC algorithm.
$b_{F,r}=1$ if the trace $r$ satisfies $aUb$ and $0$ otherwise. Please see Fig~\ref{trace1} for an illustrative example.
\end{enumerate}

Since the Bouquet algorithm obtains the exact probability for the flower, while the SMC algorithm approximates it, and the accuracy for all other classes
of traces is identical, the accuracy of the Bouquet algorithm is greater than that of the SMC algorithm for $N$ samples. The theorem follows.

\end{proof}

\begin{remark}

A trace in the Bouquet algorithm is of length at most $maxPathLength$. In a trace, the algorithm performs at every state, with probability $rProb$, 
a reachability test costing $O(k)$, and potentially a numerical model checking procedure (if a flower is encountered), of cost $(NMC)_k$. 
Thus, with for $N_B$ such samples, 
the worst case time complexity of the Bouquet algorithm is 

$O(N_B(rProb).(maxPathLength).k.(NMC)_k)$. 

In practice, the number $N_B$ turns out to be a fraction of $N_s$, the number of samples that need to be generated for the same accuracy by the 
SMC algorithm. 
Experimental evidence suggests that the running time of the Bouquet algorithm is faster for the same 
accuracy (please see section 4).
\end{remark}

\begin{figure}[ht]
    \centering
    \includegraphics[width=0.7\linewidth]{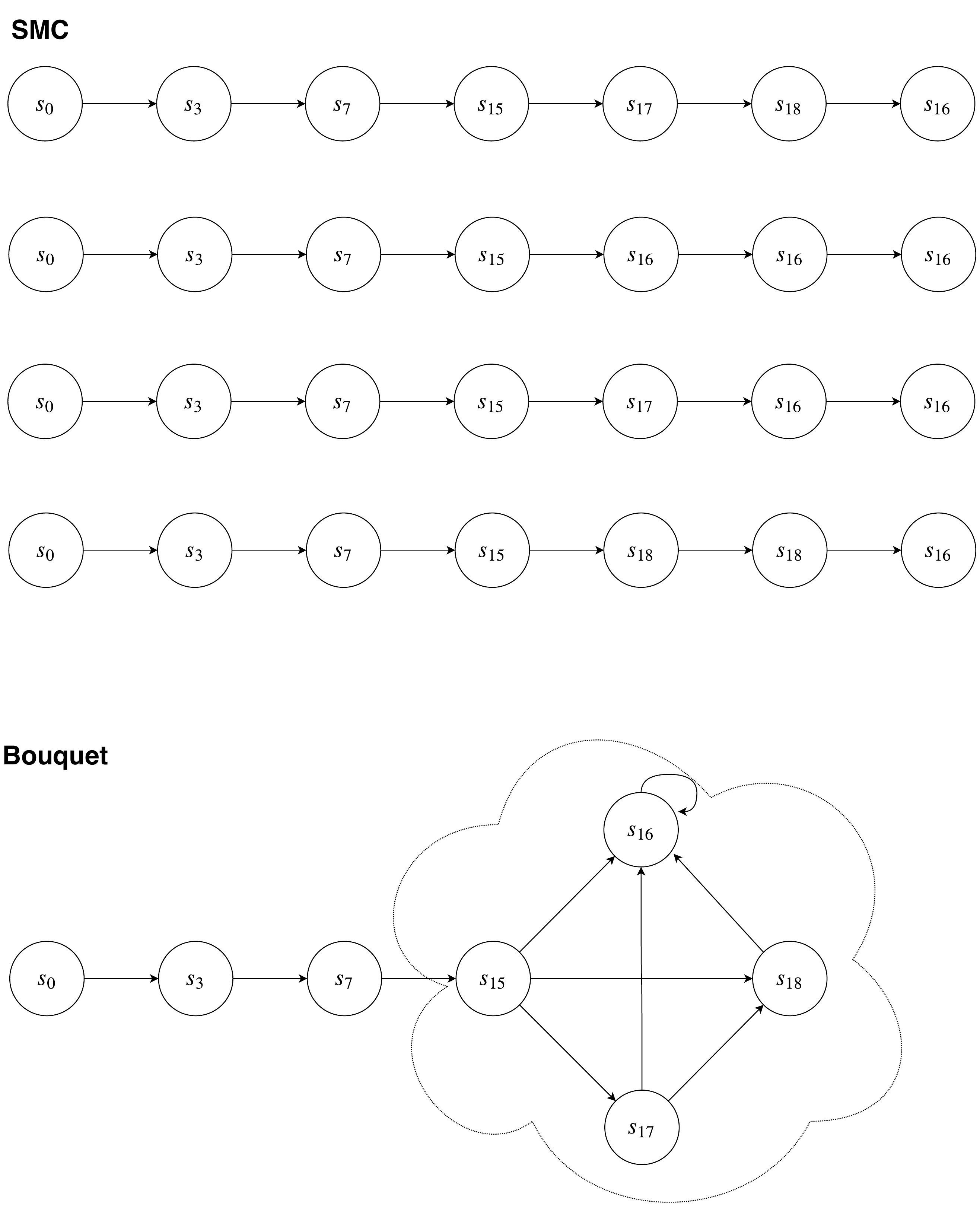}
    \caption{Trace executions for SMC and Bouquet algorithm. Note that $s_{15}$ is the starting state $s_F$ for the flower.}
    \label{trace1}
\end{figure}

\subsection{Savings in I/O Operations}
A potential application of this algorithm is in reducing page-swaps while model checking extremely large DTMCs. In what follows we call the fast,
solid-state based memory as \emph{RAM} and slower memory as the \emph{disk}. Suppose the RAM is $O(d^c)$ bits large for some constant $c$, 
while the disk is $D$ bits large, where $D>>d$. Further, suppose the DTMC is so large that the neighborhood of a node occupies $\Omega(d)$ bits, in some representation. Every time a neighbor
has to be chosen while generating a sample trace, the entire neighborhood has to be retrieved from the disk to the RAM. Therefore, for a sample of length $l$, we need to access
at least $l$ neighborhoods. For $N$ runs, we need $Nl_{avg}$ disk-RAM I/O operations, where $l_{avg}$ is the average sample length. For a completely annotated DTMC in the Bouquet algorithm, 
if the number of samples that involve a flower is $N'$ and the average length of the stalk is $l'_{avg}$, the number of I/O operations for such samples is $N'l'_{avg}$.
This results in an average saving of $Nl_{avg}-(N-N')l_{avg}-N'l'_{avg}-k=N'(l_{avg}-l'_{avg})-k$ I/O operations, where $k$ is the 
size of the flowers in the DTMC. 
The $k$ additional I/O transfers are to construct the flower. Since in this discussion we are concerned with I/O operations, we ignore the 
overhead of NMC on the flower; we assume that the NMC on the flower takes place in the RAM, which has size $O(d^c)$.

\section{Implementation and Results}\label{experiments}

We implemented the Bouquet algorithm, discussed in Section \ref{bouquet}
as a Java tool to compare its performance with standard Numerical and 
Statistical model checking algorithms%~(\url{https://github.com/shirajarora/bouquetUntil})
. We implemented the algorithm discussed in 
\cite{HanssonJ94, ModelCheck} for numerical model checking module in the bouquet algorithm.
The number of samples needed for statistical model checking is calculated using
Chernoff-Hoeffding bound~\cite{hoeffding1994,herault2004}. We empirically observed that lesser 
number of samples are needed for achieving same approximation $\epsilon$ and confidence $\delta$ 
using Bouquet algorithm, by a factor of $ \sim 0.7$. In other words, if SMC algorithms need $N_S$ samples 
to achieve ($\epsilon,\delta$), Bouquet algorithm needs $N_B\simeq 0.7N_S$ samples. 

We take as inputs sparse Markov chains with density $\rho$ and $n$ states. For a fixed value of $n$ and $\rho$, 
we randomly generated 20 different Markov chains with varying transition probability matrices.
We repeat this for different values of $n$ and $\rho$. In all the experiments, we use 
$k=\sqrt{n}$ and $rProb=0.01$.

We performed two types of experiments using the Bouquet algorithm. First, we take as an input a 
pre-annotated Markov chain wherein the number of reachable states from each state in a Markov
chain $M$ is known beforehand. This saves the bouquet algorithm the cost of 
computing reachability. For such pre-annotated Markov chains, the computation cost is 
mainly due to the numerical model checking of the flower Markov chains. 

In the second of set of experiments, the Bouquet algorithm calculates the reachability for each state on the fly 
and stores it for future visits to the state. The computation cost in these experiments is due to 
finding the number of reachable states as well as from the NMC of flower Markov chains.
 In these experiments, we observed that the first few samples
in Bouquet algorithm are more expensive in comparison to the samples executed towards the end. This is
not surprising because  more reachability information is available towards the end of simulation.  

Figure \ref{result1} shows the average time taken by the Bouquet algorithm with and without annotation of
the reachability for different number of states in the Markov chain. The density $\rho$ of these graphs is $0.05$.
 We also compare this to the time taken by SMC algorithm. We took 15 batches of 1000 sample runs for the SMC algorithm whereas
700 samples for both versions of Bouquet algorithm. As can be seen, when the graph is annotated completely (either exclusively
for this unbounded until query, or due to evaluations of previous queries), the Bouquet algorithm performs better than the SMC algorithm.
However, as expected, when the graph is not pre-annotated, it performs slightly worse than SMC.

\begin{figure}[ht!]
    \centering
    \includegraphics[width=0.8\linewidth]{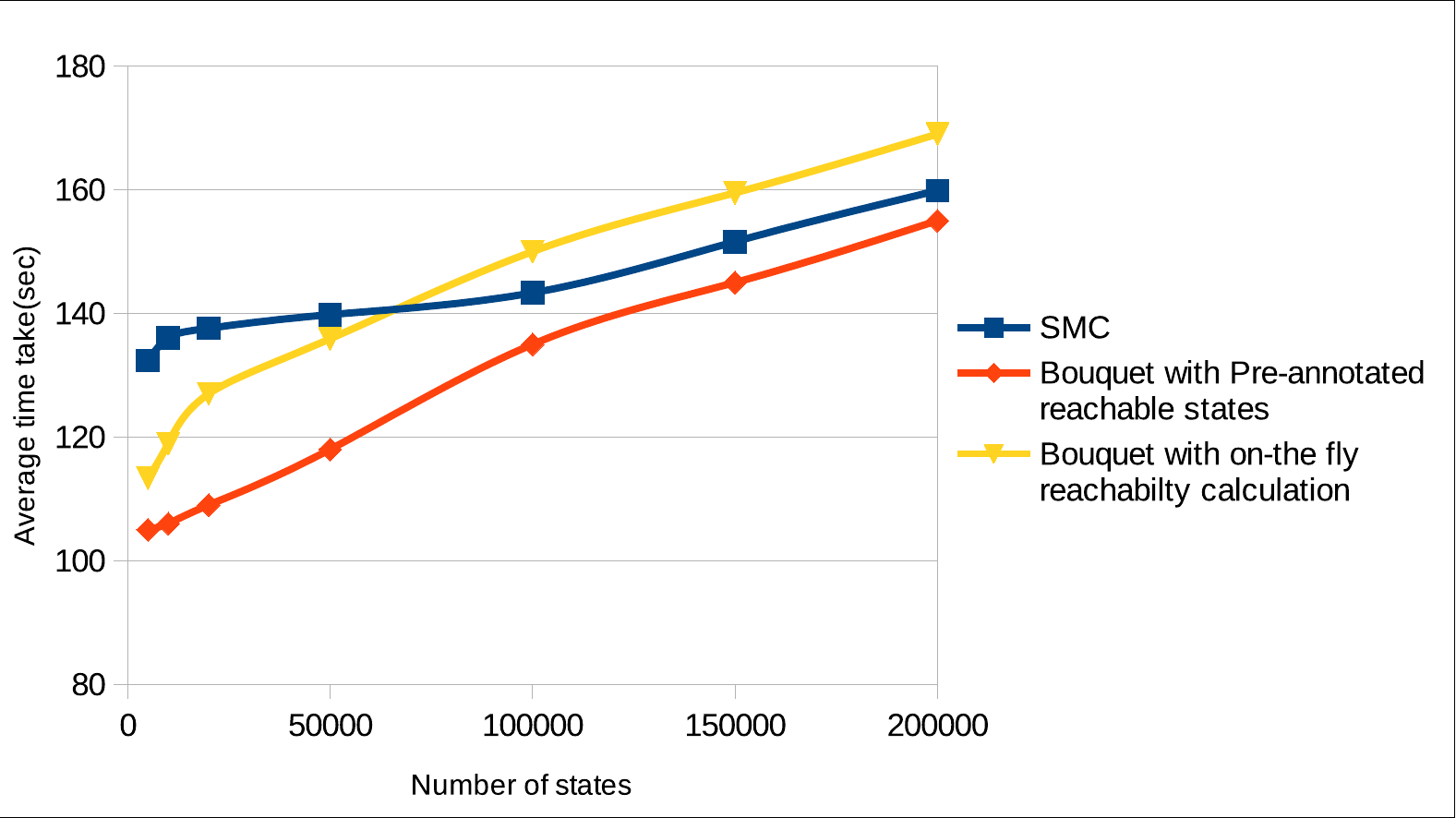}
    \caption{Comparison of average time taken by different algorithms for density $\rho=0.05$ }
    \label{result1}
\end{figure}
 
%It is important to note that while the Bouquet algorithm takes comparable (sometimes even more) time than
%the  SMC algorithm, it provides better accuracy. The numerical computation for verifying the flower Markov 
%chains decreases the probability of incorrect result significantly, as opposed to the sampling within the flower done 
%by statistical model checking.

In the case where reachable states are identified on the fly, the reachability results 
from verification of one unbounded formula can be directly reused 
during the verification of another unbounded until formula. This will be useful if the user wants to
verify multiple unbounded until formulas for the same system over time. Figure \ref{result2} shows the comparison of
average time taken by the Bouquet algorithm to consecutively verify three different unbounded until 
formulas on the same model. It is evident from the plot that as the graph gets progressively annotated, the time taken to check 
the unbounded until formula drops.

\begin{figure}[ht]
    \centering
    \includegraphics[width=0.7\linewidth]{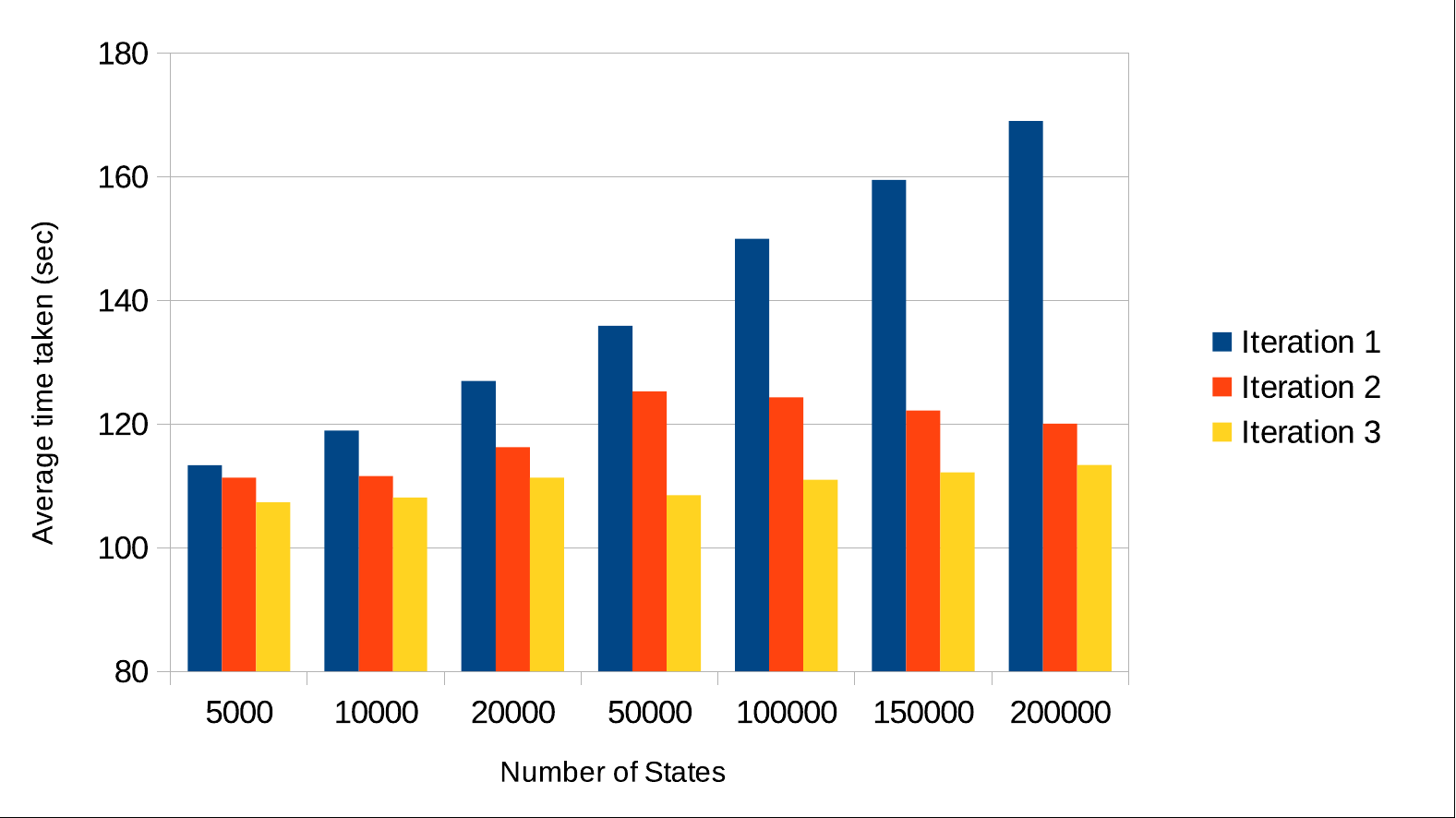}
    \caption{Average time taken to consecutively verify different unbounded until formulas}
    \label{result2}
\end{figure}

We compare the performance of on-the-fly bouquet algorithm with statistical model checking for different values of $\epsilon$, the 
approximation parameter. Figure~\ref{result4} shows the results for this experiment. We see that with increase in the desired accuracy, the bouquet 
algorithm outperforms the statistical model checking algorithm significantly. The reason is that the overhead of additional samples that the SMC
algorithm needs to generate for achieving greater accuracy is greater than the numerical model checking cost in the Bouquet algorithm.

\begin{figure}[ht]
    \centering
    \includegraphics[width=0.7\linewidth]{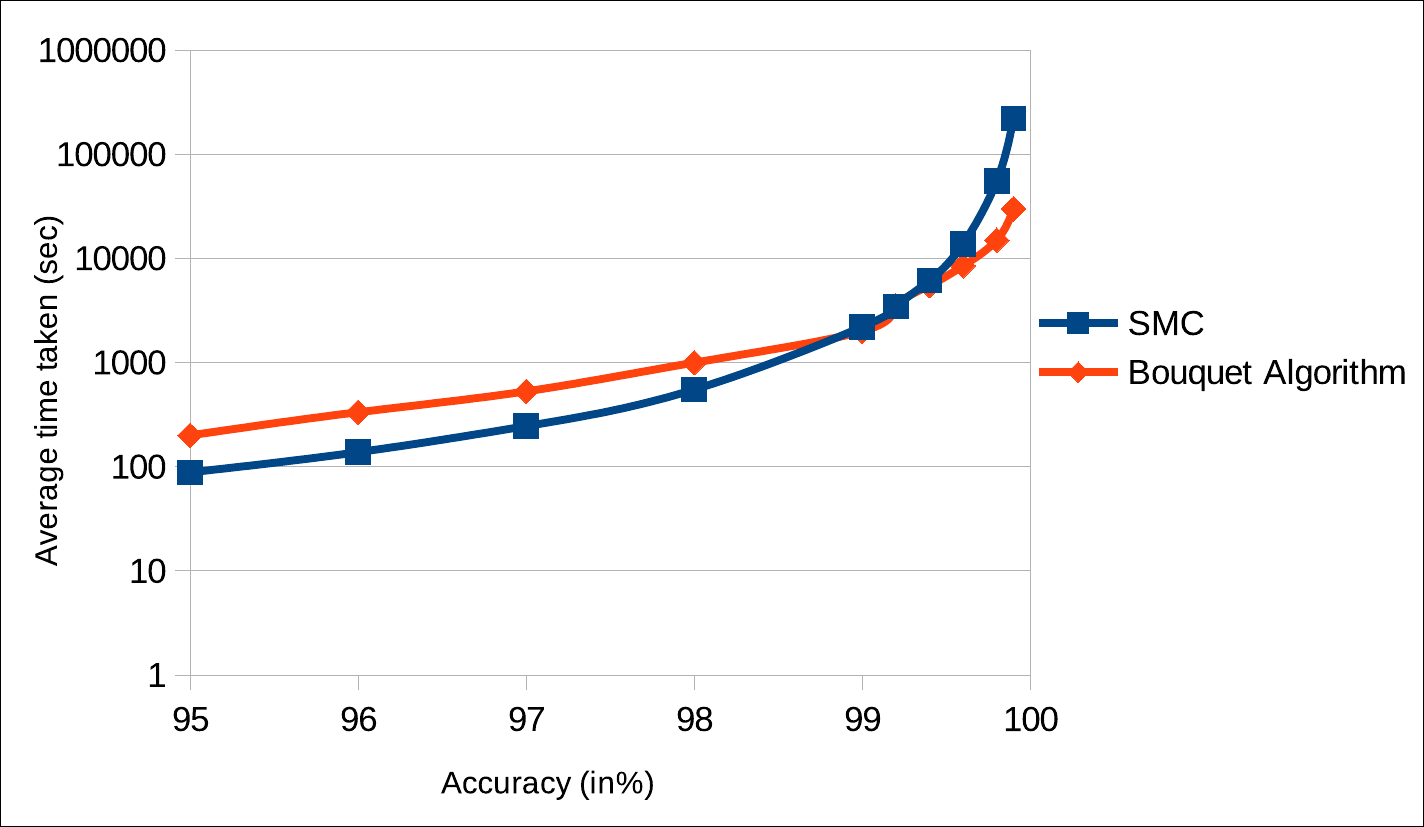}
    \caption{Accuracy vs Time in Bouquet Algorithm}
    \label{result4}
\end{figure}

In the case of fully annotated but dense graphs, the Bouquet algorithm converges to the SMC algorithm, because of the abundance of long sample paths that do
not end in small sized flowers. In all cases, however, both SMC and Bouquet algorithms outperform NMC in terms of running time.
Indeed, we observed that the performance of Bouquet algorithm where reachability is annotated on the fly
improves with the increase in sparsity of the underlying directed graph of the Markov chain. Figure \ref{result3} 
illustrates the average time taken to verify Markov chain with $10^5$ states for different densities.

\begin{figure}[ht]
    \centering
    \includegraphics[width=0.7\linewidth]{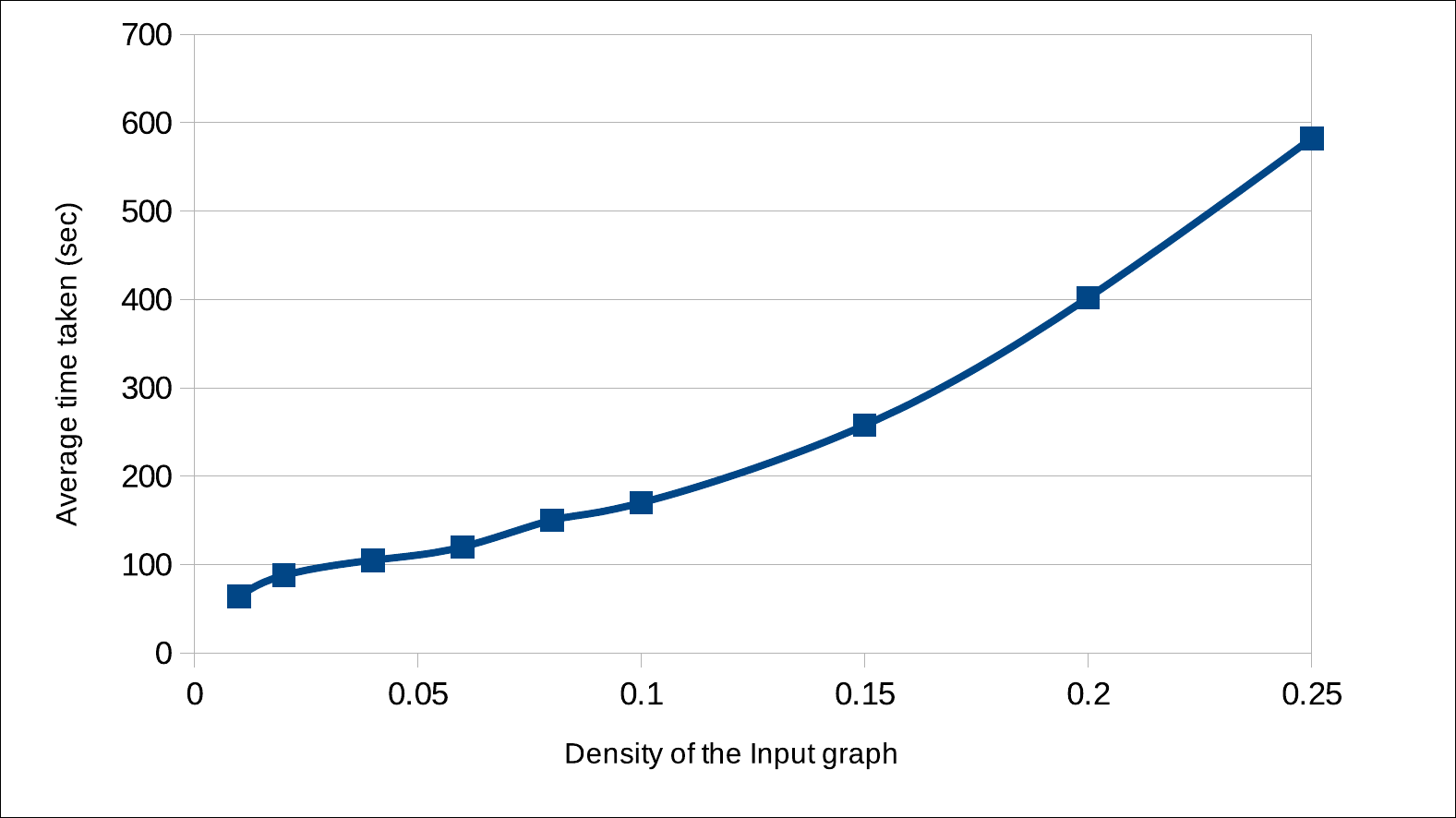}
    \caption{Average time taken by Bouquet algorithm for varying density $\rho$ }
    \label{result3}
\end{figure}

\section{Conclusions and Future work}\label{conclusion}
We discussed a new hybrid algorithm for verifying the unbounded until fragment of temporal logics, using numerical and statistical model checking 
and graph algorithms and demonstrated its effectiveness on sparse DTMCs. 
In particular, we give empirical evidence for improved performance of this approach over the standard statistical model checking algorithm.

As we mention is section 3.4, we believe that this technique could be of immense use when we seek to reduce I/O operations for DTMCs 
explicitly stored on a slow but inexpensive memory. It remains to confirm this conjecture experimentally over different memory architectures. 
Finally, we plan to integrate this into the PRISM model checker and address practical model checking problems. 

\bibliographystyle{splncs04}
\bibliography{bib}

\begin{thebibliography}{10}
\providecommand{\url}[1]{\texttt{#1}}
\providecommand{\urlprefix}{URL }
\providecommand{\doi}[1]{https://doi.org/#1}

\bibitem{agha2018}
Agha, G., Palmskog, K.: A survey of statistical model checking. ACM
  Transactions on Modeling and Computer Simulation (TOMACS)  \textbf{28}(1), ~6
  (2018)

\bibitem{ModelCheck}
Baier, C., Katoen, J.P.: Principles of Model Checking (Representation and Mind
  Series). The MIT Press (2008)

\bibitem{basu2009}
Basu, S., Ghosh, A.P., He, R.: Approximate model checking of pctl involving
  unbounded path properties. In: International Conference on Formal Engineering
  Methods. pp. 326--346. Springer (2009)

\bibitem{courcoubetis88}
Courcoubetis, C., Yannakakis, M.: Verifying temporal properties of finite-state
  probabilistic programs. In: [Proceedings 1988] 29th Annual Symposium on
  Foundations of Computer Science. pp. 338--345. IEEE (1988)

\bibitem{daca2017}
Daca, P., Henzinger, T.A., K{\v{r}}et{\'\i}nsk{\`y}, J., Petrov, T.: Faster
  statistical model checking for unbounded temporal properties. ACM
  Transactions on Computational Logic (TOCL)  \textbf{18}(2), ~12 (2017)

\bibitem{rabih2009}
El~Rabih, D., Pekergin, N.: Statistical model checking using perfect
  simulation. In: International Symposium on Automated Technology for
  Verification and Analysis. pp. 120--134. Springer (2009)

\bibitem{HanssonJ94}
Hansson, H., Jonsson, B.: A logic for reasoning about time and reliability.
  Formal Aspects of Computing  \textbf{6}(5),  512--535 (1994).
  \doi{10.1007/BF01211866}, \url{http://dx.doi.org/10.1007/BF01211866}

\bibitem{he2010}
He, R., Jennings, P., Basu, S., Ghosh, A.P., Wu, H.: A bounded statistical
  approach for model checking of unbounded until properties. In: Proceedings of
  the IEEE/ACM international conference on Automated software engineering. pp.
  225--234. ACM (2010)

\bibitem{herault2004}
H{\'e}rault, T., Lassaigne, R., Magniette, F., Peyronnet, S.: Approximate
  probabilistic model checking. In: International Workshop on Verification,
  Model Checking, and Abstract Interpretation. pp. 73--84. Springer (2004)

\bibitem{hoeffding1994}
Hoeffding, W.: Probability inequalities for sums of bounded random variables.
  In: The Collected Works of Wassily Hoeffding, pp. 409--426. Springer (1994)

\bibitem{lassaigne2015}
Lassaigne, R., Peyronnet, S.: Approximate planning and verification for large
  markov decision processes. International Journal on Software Tools for
  Technology Transfer  \textbf{17}(4),  457--467 (2015)

\bibitem{nimal2010}
Nimal, V.: Statistical approaches for probabilistic model checking. Ph.D.
  thesis, University of Oxford (2010)

\bibitem{roohi2015}
Roohi, N., Viswanathan, M.: Statistical model checking for unbounded until
  formulas. International Journal on Software Tools for Technology Transfer
  \textbf{17}(4),  417--427 (2015)

\bibitem{sen2004}
Sen, K., Viswanathan, M., Agha, G.: Statistical model checking of black-box
  probabilistic systems. In: International Conference on Computer Aided
  Verification. pp. 202--215. Springer (2004)

\bibitem{sen05}
Sen, K., Viswanathan, M., Agha, G.: On statistical model checking of stochastic
  systems. In: International Conference on Computer Aided Verification. pp.
  266--280. Springer (2005)

\bibitem{NumVsStat}
Younes, H.L.S., Kwiatkowska, M.Z., Norman, G., Parker, D.: Numerical vs.
  statistical probabilistic model checking: An empirical study. In: Tools and
  Algorithms for the Construction and Analysis of Systems, 10th Intl. Conf.,
  {TACAS} 2004, Held as Part of the Joint European Conf.s on Theory and
  Practice of Software, {ETAPS} 2004, Barcelona, Spain, March 29 - April 2,
  2004, Proc. pp. 46--60 (2004)

\bibitem{Younes}
Younes, H.L.S., Simmons, R.G.: Probabilistic verification of discrete event
  systems using acceptance sampling. In: Computer Aided Verification, 14th
  Intl. Conf., {CAV} 2002,Copenhagen, Denmark, July 27-31, 2002, Proc. pp.
  223--235 (2002)

\bibitem{younes2010}
Younes, H.L., Clarke, E.M., Zuliani, P.: Statistical verification of
  probabilistic properties with unbounded until. In: Brazilian Symposium on
  Formal Methods. pp. 144--160. Springer (2010)

\end{thebibliography}
\end{document}